\documentclass{ifacconf}

\usepackage{graphicx}      
\usepackage{natbib}        
\usepackage{amsmath}
\usepackage{amssymb}
\begin{document}

    \begin{frontmatter}

        \title{Resource-Constrained Shortest Path with Polytopic Reset Sets\thanksref{footnoteinfo}} 

        \thanks[footnoteinfo]{This work was supported by the Office of Naval Research (ONR) under grant N00014-25-1-2519.}

        \author[First]{Khaled Surur} 
        \author[Second]{and Melkior Ornik}

        \address[First]{Department of Electrical and Computer Engineering, University of Illinois Urbana-Champaign, Urbana, IL 61801 USA}
        \address[Second]{Department of Aerospace Engineering, University of Illinois Urbana-Champaign, Urbana, IL 61801 USA \\ (e-mail: surur2@illinois.edu, mornik@illinois.edu)}

        \begin{abstract}
            This paper investigates the problem of computing the shortest path between two states under resource constraints in environments with resource-replenishment regions. Namely, the length of the path is limited by a budget that can be restored within polytopic replenishment regions. We show that the optimal path in this problem exhibits a distinct geometric structure: it consists of straight-line segments, changes direction at the edges of replenishment regions, and enters a sequence of replenishment regions in a feasible-connected space. We propose an approach to solve the continuous problem in two steps: using a graph-based approach, followed by convex programming. First, we define a graph whose nodes are possible waypoints of feasible paths, and the edges are the Euclidean distances between these nodes. To obtain a discrete set of nodes that ensure a feasible and near-optimal solution, we utilize a wavefront algorithm. With a sufficiently small spacing between wavefronts, the solution of the shortest path problem on this graph yields the optimal sequence of polytopes to visit. Next, we use convex optimization on this sequence of polytopes to refine the solution to optimality. A numerical experiment is presented to demonstrate the effectiveness of the approach.
        \end{abstract}


    \end{frontmatter}
    

    \section{Introduction}

        Resource-constrained shortest path problems arise in many engineering applications, including autonomous systems, optimal control, and operation research, where systems have limited budget on resources that are depleted throughout the operation. Application examples where resource-constrained path planning is significant include persistent aerial surveillance \citep{Zuo:2020}, transportation \citep{Tu:2020}, computer networks \citep{Chemodanov:2019}, and cooperative multi-agent systems \citep{Mondal:2025}. In robotics and mobile systems, a typical problem is to compute the shortest or fastest path to a target subject to resource constraints. Variants of this problem include minimizing exposure to adversaries by treating visibility time as a limited resource \citep{Kumar:2010} and incorporating localization uncertainty as a budget constraint in navigating environments denied by the global positioning system (GPS) \citep{Bopardikar:2015}.

        In realistic settings, the state transitions of a system are constrained by a resource budget (such as fuel, time, or capacity) which are depleted or diminished as control effort is exerted or the configuration of the system changes. Consequently, resources in the system may have to be renewed to prevent exhaustion.

        A large body of work addresses the resource-constrained shortest path problem using discrete graph-based models, where system states are represented as nodes. There exist heuristic tools to solve the shortest graph problems. For example, a refueling path planning problem for vehicles has been investigated using a modified \(A^*\) algorithm \citep{Zhao:2025}. Another line of work formulates routing problems using mixed-integer linear programming (MILP). For example, an energy-aware MILP has been developed for combinatorial optimization over routes to solve the optimal routing for an autonomous aircraft in the presence of refueling depots \citep{Sundar:2014}. Although these approaches are effective, they typically assume that resource replenishment occurs at discrete locations.

        An alternative approach to solving resource-constrained shortest path problems can be formulated in a continuous optimal control framework. For example, dynamic programming methods have been used to compute optimal policies under budget constraints with reset mechanisms by evaluating approximate cost functions \citep{Takei:2015}, while control-theoretic approaches such as control barrier functions have been used for autonomous systems with energy constraints to enforce resource feasibility and task persistency \citep{Notomista:2021}. However, these methods often suffer from high computational complexity or rely on discretization schemes that may limit scalability.

        Motion planning methods based on computational geometry provide additional tools to address shortest path problems, particularly in environments with geometric constraints \citep{deBerg:2013}. More recently, graphs of convex sets (GCSs) have emerged as a powerful framework for modeling continuous path planning problems, where nodes correspond to convex regions and edge costs are defined over continuous decision variables \citep{Kurtz:2023, Marcucci:2024, Garg:2025}. This framework enables the incorporation of complex geometric and dynamical constraints within a convex optimization setting to solve shortest path problems.

        Despite these advances, existing approaches either rely on discrete refueling nodes or lack scalable methods for handling resource replenishment over continuous settings. While GCS provides a powerful framework for continuous paths, the integration of cumulative resource constraints with budget reset mechanisms remains an open problem. In this paper, we address this gap by considering a generalized case of a continuous resource-constrained shortest path problem. Rather than discrete nodes or depots, resource replenishment occurs over convex regions in the state space, thus capturing a broader class of realistic scenarios. The contributions of this paper are twofold.
        \begin{itemize}
            \item We formulate the resource-constrained shortest path problem with replenishment sets under a GCS framework. In particular, we exploit key geometric structure and length-based budget constraint to reduce the problem to an optimal path computation over a region-distance graph constructed via pairwise convex set distances.
            \item We utilize a wavefront algorithm that generates candidate nodes to narrow the set of feasible paths. We then solve the problem in two steps: (i) we solve the shortest graph generated by the candidate nodes using Dijkstra method to obtain the optimal polytopic sequence the path will visit, then (ii) we apply convex programming on the sequence obtained to calculate exact solutions.
        \end{itemize}
        The resulting approach provides a computationally efficient solution with provable optimality.

        \textbf{Notations.} Throughout this paper, the mathematical operator \(\|\cdot\|\) denotes the Euclidean norm (2-norm). For a closed set \(\mathcal{S}\), the expression \(\partial \mathcal{S}\) denotes its boundary, with \(\partial \mathcal{S} \subseteq \mathcal{S}\). The operator \(\oplus\) denotes the Minkowski sum of two sets such that, for given two sets \(\mathcal{S}_1\) and \(\mathcal{S}_2\), \(\mathcal{S}_1 \oplus \mathcal{S}_2 = \{a + b : a \in \mathcal{S}_1, b \in \mathcal{S}_2\}\).

    \section{Problem Formulation}
    
        Consider an initial point \(x_S \in \mathbb{R}^n\) and a terminal point \(x_E \in \mathbb{R}^n\). Let \(f:[a,b] \rightarrow \mathbb{R}^n\) be a piecewise differentiable path such that \(f(a) = x_S\) and \(f(b) = x_E\). The length of the path \(f\) is defined as
        \begin{equation}
            L(f) = \int_a^b \|f'(\tau)\| \ \mathrm{d} \tau
        \end{equation}
        where \(f'(\tau)\) denotes the tangent vector of the path. Let \(\{\mathcal{P}_1, \mathcal{P}_2, \dots, \mathcal{P}_m\} \subset \mathbb{R}^n\) be a collection of pairwise disjoint polytopes, i.e., \(\mathcal{P}_i \cap \mathcal{P}_j = \emptyset\) for all \(i \neq j\). Polytopes are defined by the half-space representation in which there exist a matrix \(H \in \mathbb{R}^{k \times n}\) and a vector \(h \in \mathbb{R}^k\) (where \(k \geq n+1\)) such that
        \begin{equation}
                \mathcal{P} = \{x \in \mathbb{R}^n : H \, x \leq h\}
        \end{equation}
        and \(\mathcal{P}\) is bounded. By definition, polytopes are considered convex and closed, and they are referred to as replenishment regions throughout this paper. Let \(Q > 0\) denote the maximum allowable travel distance (budget).

        Let the system be modeled as a vehicle traveling on a path as it expends a resource amount equal to the length of the path traveled, thus the vehicle can travel at most a distance \(Q\) without replenishing. The vehicle's initial resource budget is set to its full capacity \(Q\), and upon entering or residing within a polytope \(\mathcal{P}_i\), the budget is ``reset'' and is fully replenished to \(Q\). Consequently, the vehicle's motion outside replenishment regions must satisfy the distance constraint, while any path segment within the regions is not constrained by the budget. The system dynamics is thus a hybrid dynamical system.

        We wish to find the shortest path between two locations, subject to the vehicle always maintaining a nonnegative resource level while traveling to a target destination. In this setting, we formulate the resource-constrained shortest path problem as follows.

        \textit{Problem 1.} The objective is to minimize \(L(f)\) over all feasible paths \(f:[0,T] \rightarrow \mathbb{R}^n\) from \(x_S\) to \(x_E\) in the presence of \(\{\mathcal{P}_1, \mathcal{P}_2, \dots, \mathcal{P}_m\}\), where the resource level \(r\) remains nonnegative, and the budget constraint \(Q\) is satisfied. Consequently, the optimization problem is described as
        \begin{subequations}
            \label{eq:MinimumPath}
            \begin{align}
                \min_{f(t)} \quad & L(f) \\
                \text{s.t.} \quad & r(t) \geq 0, \ \forall t \in [0,T]
            \end{align}
        \end{subequations}
        where \(f(0) = x_S\), \(f(T) = x_E\), \(r(0) = Q\), \(r'(t) = -1\) whenever \(f(t) \notin \bigcup_{i=1}^m \mathcal{P}_i\), and \(r(t) = Q\), \(r'(t) = 0\) whenever \(f(t) \in \bigcup_{i=1}^m \mathcal{P}_i\).
        
        As a result, the optimal solution depends on the selection of a feasible path \(f\) that minimizes \(L(f)\) and satisfies the distance budget \(Q\).     

    \section{Methodology}

        In this section, we first define key characteristics that describe the optimal path in Problem~1. We solve Problem~1 in two stages: a resource-constrained graph search identifies the ordered sequence of replenishment regions visited by the shortest feasible path, and convex programming refines the path within that sequence to optimality, reducing Problem~1 to a (GCS) problem.
        
        \subsection{Properties of the Optimal Path}

            Without resource constraints, the shortest path between two arbitrary points \(x_a\) and \(x_b\) is obviously a straight line, with length that equals \(\|x_a - x_b\|\). This can be proven easily using the triangle inequality. The subsequent result obviously follows.
            \begin{lem}
                For a given ordered set of points \(p_0,p_1,\dots,p_k\), defined in \(\mathbb{R}^n\) space, the minimum path \(f\) that visits the points in order consists of piecewise line segments and has length \(L^*(f) = \sum_{i=0}^{k-1} \|p_i - p_{i+1}\|\).
            \end{lem}
            
            Now suppose that parts of \(L(f)\) are limited by a length budget \(Q\), which is reset whenever \(f\) enters or resides within some \(\mathcal{P}_i\). Naturally, if \(\|x_S - x_E\| \leq Q\), the optimal path is a straight line from \(x_S\) to \(x_E\), and \(L^*(f) = \|x_S - x_E\|\). We now formally define the general case of the feasible resource-constrained path as follows.

            \begin{lem}
                Let \(f:[0,T] \rightarrow \mathbb{R}^n\) be a path with \(f(0) = x_S\), \(f(T) = x_E\), \(\{\mathcal{P}_1, \mathcal{P}_2, \dots, \mathcal{P}_m\}\), and \(Q > 0\) given. The path \(f\) is feasible if there exist \( p_0 < p_1 < \dots < p_k\) such that \(p_0 = 0\), \(p_k = T\), \(f(p_j) \in \bigcup_{i = 1}^m \partial\mathcal{P}_i, \forall j=1,2,\dots,k-1\), and, for every \(s = 1,2,\dots,k\), the path segment \(f_s \subset f\) between \(f(p_{s-1})\) and \(f(p_{s})\) satisfies \(L(f_{s}) \leq Q\) whenever this segment does not lie entirely within some replenishment region, i.e., whenever \(f_s(\rho) \notin \bigcup_{i=1}^m \mathcal{P}_i\) for any \(\rho \in (p_{s-1},p_{s})\).
            \end{lem}

            \begin{pf}
                Since the budget resets to \(Q\) only while \(f(\rho)\) remains inside some \(\mathcal{P}_i\), any segment \(f_s\) that leaves every \(\mathcal{P}_i\) depletes the budget at unit rate and must satisfy \(L(f_s)\le Q\) for \(f\) to remain feasible. \hfill \(\blacksquare\)
            \end{pf}

            A simple proof sketch is shown in Fig.~\ref{fig:GeneralConstrainedPath} for the case of \(k=3\). In this example, the total length of path segments \(f_1\), connecting the points \(f(p_0)\) and \(f(p_1)\), and \(f_3\), connecting the points \(f(p_2)\) and \(f(p_3)\), must be less than or equal to \(Q\) for feasibility since they are outside \(\mathcal{P}\). Once the path enters \(\mathcal{P}\) through \(f(p_1)\), the budget constraint \(Q\) resets and remains so as long the path exists within \(\mathcal{P}\). Consequently, \(L(f_2)\) is not limited by \(Q\) and \(f_2\) can move anywhere in the replenishment region.
            
            \begin{figure}[hb]
                \begin{center}
                    \includegraphics[width=8.4cm]{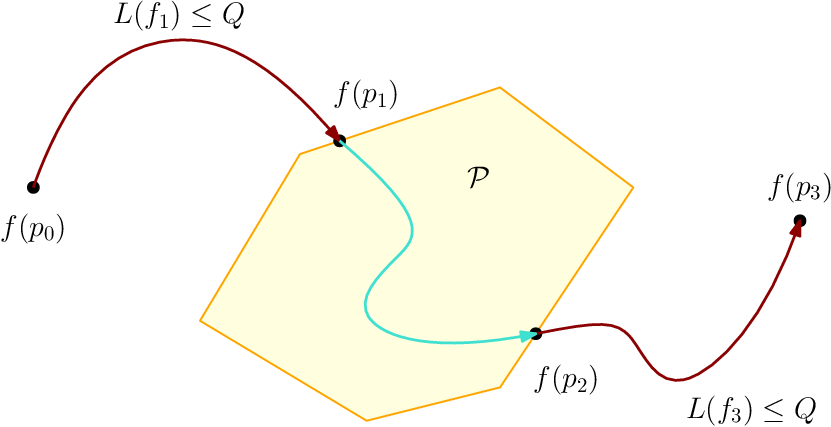}    
                    \caption{An example of a feasible resource-constrained path in the presence of replenishment regions as described by Lemma~2.} 
                    \label{fig:GeneralConstrainedPath}
                \end{center}
            \end{figure}

            Since \(L(f)\) is limited by \(Q\), the farthest extent the path \(f\) can reach without replenishing its budget can be bounded by a ball with a radius \(Q\). A ball is defined as 
            \begin{equation}
                \mathcal{B}(c;r) := \{ p \in \mathbb{R}^n : \|p - c\| \leq r \}
            \end{equation}
            where \(c\) and \(r\) are the center of the ball and the radius, respectively. Consequently, \(\mathcal{B}(c;r)\) contains all budget-constrained paths \(f_i\) that start at \(c\) and \(0 \leq L(f_i) \leq Q\). By Lemma~2, the following is inferred.

            \begin{cor}
                Consider the sets \(\mathcal{B}_S(x_S;r_S)\), \(\mathcal{B}_E(x_E;Q)\), and \(\mathcal{B}_0(0;Q)\), with \(0 < r_S \leq Q\). Let \(\mathcal{D}_{i,j} = \mathcal{P}_i \cap (\mathcal{P}_j \oplus \mathcal{B}_0)\), \(\mathcal{D}_{i,0} = \mathcal{B}_S \cap \mathcal{P}_i\), \(\mathcal{D}_{i,m+1} = \mathcal{B}_E \cap \mathcal{P}_i\), \(\forall i,j=1,2,\dots,m\), where \(i \neq j\), \(\mathcal{E}_{i,k} = \partial\mathcal{P}_i \cap \partial\mathcal{D}_{i,k}\), \(\forall k=0,1,\dots,m+1\), where \(i \neq k\), be subspaces in \(\mathbb{R}^n\). The resource-constrained path \(f\) from \(x_S\) to \(x_E\) is feasible only if \(\bigcup_{i=1}^m \mathcal{D}_{i,0} \neq \emptyset\), \(\bigcup_{i=1}^m \mathcal{D}_{i,m+1} \neq \emptyset\), and there exist nonempty sets
                \begin{equation}
                    \mathcal{A} = \{\mathcal{P}_i : \exists \mathcal{D}_{i,j} \subset \mathcal{P}_i\}
                \end{equation}
                for all \(i = 1,2,\dots,m\) and \(j = 0,1,\dots,m+1\), where \(j \neq i\), and
                \begin{align}
                    \mathcal{M}_S &= \bigcup_{\substack{1 \leq i \leq m}} \ \bigcup_{t \in [0,1]} \big[ t \, x_S \oplus (1-t) \, \mathcal{E}_{i,0} \big]
                    \\
                    \mathcal{M}_E &= \bigcup_{\substack{1 \leq i \leq m}} \ \bigcup_{t \in [0,1]} \big[ t \, x_E \oplus (1-t) \, \mathcal{E}_{i,m+1} \big]
                    \\
                    \mathcal{M}_P &= \bigcup_{\substack{1 \leq i,j \leq m \\ i \neq j, \ (i) < (j)}} \ \bigcup_{t \in [0,1]} \big[ t \, \mathcal{E}_{i,j} \oplus (1-t) \, \mathcal{E}_{j,i} \big]
                \end{align}
                where \((i)<(j)\) denotes lexicographic order, used only to avoid duplicating each unordered pair, such that the subspace \(\mathcal{R} = \mathcal{A} \cup \mathcal{M}_P \cup \mathcal{M}_S \cup \mathcal{M}_E\), defined in \(\mathbb{R}^n\), is a path-connected space and \(f \subset \mathcal{R}\).
            \end{cor}

            \begin{pf}
                Both \(\bigcup_{i=1}^m \mathcal{D}_{i,0} \neq \emptyset\) and \(\bigcup_{i=1}^m \mathcal{D}_{i,m+1} \neq \emptyset\) are necessary to construct a budget-constrained path from \(x_S\) to some \(\mathcal{P}_i\) and from some \(\mathcal{P}_i\) to \(x_E\), respectively. In addition, if \(\mathcal{D}_{i,j}\) exists for some \(i\) and \(j\), where \(i \neq j\), then \(\mathcal{D}_{j,i}\) also exists by definition, and an inter-regional space between \(\mathcal{P}_i\) and \(\mathcal{P}_j\) can be formed. These interstitial spaces, expressed as \(t \,\mathcal{E}_{i,j} \oplus (1-t) \, \mathcal{E}_{j,i}\) with \(t \in [0,1]\), contain feasible paths that meet the budget constraint outside the replenishment regions. If enough such \(\mathcal{D}_{i,j}\) are nonempty that the union of the corresponding \(\mathcal{P}_i\) and the interstitial spaces forms a connected space containing \(x_S\) and \(x_E\), then a feasible path \(f\) exists within this set satisfying the budget constraint. \hfill \(\blacksquare\)
            \end{pf}

            A proof sketch is shown in Fig.~\ref{fig:FeasiblePath}. Here, the spaces \(\mathcal{D}_{1,0}\), \(\mathcal{D}_{1,2}\), \(\mathcal{D}_{2,1}\), and \(\mathcal{D}_{2,3}\) indicate the regions where a feasible path can move from \(x_S\), \(x_E\), or some \(\mathcal{P}_i\) to a replenishment region without violating the budget \(Q\). The space \(\mathcal{R}\) contains all replenishment regions with existent \(\mathcal{D}_{i,j}\) and the interstitial spaces outside them, within which it is possible to traverse while remaining within the length limit. Thus, the continuous path \(f\) is feasible only if it is entirely in \(\mathcal{R}\) where the budget can be reset without violating the length constraint \(Q\).

            \begin{figure}[hb]
                \begin{center}
                    \includegraphics[width=8.4cm]{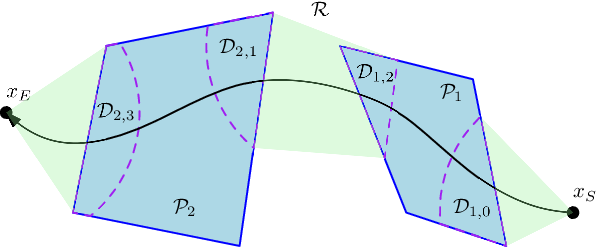}    
                    \caption{A proof sketch of corollary 3 where the feasible path \(f\) exists in the path-connected space \(\mathcal{R}\) which comprises the inter-regional spaces \(\mathcal{M}_P\), \(\mathcal{M}_S\), \(\mathcal{M}_E\)  (light green), and the polytopes in \(\mathcal{A}\) (light blue).} 
                    \label{fig:FeasiblePath}
                \end{center}
            \end{figure}

            From Lemma~2 and Corollary~3, a feasible path can be decomposed as
            \begin{equation*}
                x_S \rightarrow a_{\sigma_1} \rightarrow b_{\sigma_1} \rightarrow a_{\sigma_2} \rightarrow \dots \rightarrow b_{\sigma_k} \rightarrow x_E
            \end{equation*}
            such that \(a_i, b_i \in \mathcal{P}_i\) are entry and exit points, respectively, for some ordered sequence of visited polytopes indexed by \(\sigma = (\sigma_1, \sigma_2, \dots, \sigma_k)\) where \(\sigma_j \in \{1,2,\dots,m\}\). The optimization problem in~(\ref{eq:MinimumPath}) can be reformulated to
            \begin{subequations}
                \label{eq:MinimumPath2}
                \begin{align}
                    \min_{\substack{a_{\sigma_i}, b_{\sigma_i} \in \mathcal{P}_{\sigma_i} \\ 1 \leq i \leq k}} \quad & L(f) \\
                    \text{s.t.} \quad & l_{2i+1} \leq Q, \quad i = 0,1,\dots,k
                \end{align}
            \end{subequations}
            where the length of the decomposed path is represented as \(L(f) = l_1 + l_2 + \dots + l_{2k+1}\) such that \(l_1\) relates to the path \(x_S \rightarrow a_{\sigma_1}\), \(l_{2i}\) relates to the path \(a_{\sigma_i} \rightarrow b_{\sigma_i}\) for all \(i\), \(l_{2i+1}\) relates to the path \(b_{\sigma_i} \rightarrow a_{\sigma_{i+1}}\) for all \(i\), and \(l_{2k+1}\) relates to the path \(b_{\sigma_k} \rightarrow x_E\). For a feasible path, each inter-region segment satisfies \(l_{2i+1} \leq Q,\) for all \(i=0,1,\dots,k\). As a result, the optimal solution depends on the selection of the points \(a_i\) and \(b_i\) that both minimize \(L(f)\) and uphold the transition constraint \(Q\).

            The shortest resource-constrained path has key characteristics: (i) the path consists entirely of piecewise straight-line segments (Lemma~1), (ii) the endpoints of the line segments lie in some \(\mathcal{P}\) (Lemma~2), and (iii) the path enters a sequence of replenishment regions in a feasible-connected space (Corollary~3). The properties of the optimal path make Problem~1 equivalent to the shortest path in a GCS under a distance constraint. Taking this into account, the continuous optimization problem can be reduced to a finite graph search problem. The problem is now equivalent to finding the shortest path in a graph whose nodes and edges are the points on the boundary of the convex regions and the distances between these points, respectively.

        \subsection{Resource-Constrained Graph with Budget Reset}

            We reformulate the continuous problem of the resource-constrained shortest path into a discrete structure where the path is described by nodes and edges. Consider a graph \(\mathcal{G} := (V,E)\) where \(V\) and \(E\) represent the set of nodes and the edges between them, respectively. Given that the optimal path consists of piecewise line segments, \(V\) consists of points on each \(\partial \mathcal{P}\) that the path \(f\) may pass through, together with \(x_S\) and \(x_E\), which serve as the source and terminal nodes, respectively.
            
            Here, nodes are points that can potentially be \(a_{\sigma_i}\) or \(b_{\sigma_i}\). Since \(\partial \mathcal{P}_i\) is a continuous set, we select a discrete candidate set of nodes that will yield feasible and near-optimal solutions. From Corollary 3, the goal is to obtain candidate points that lie at \(\mathcal{E}_{i,j}\), narrowing the candidate set. To determine these nodes, we apply a wavefront point generation strategy. In this method, we generate a sequence of spheres centered at given points called seeds with radii (levels) that gradually vary at a uniform step size equal to \(\delta\) up to \(Q\). The levels are determined by \(\Delta_k = k \, Q/\delta, \forall k=1,2,\dots,\delta\). At each circular level, we calculate the intersection points of these spheres with each \(\partial \mathcal{P}_i\), which become candidate nodes. Seed points are selected such that the propagated spheres intersect with any \(\mathcal{E}\). The points \(x_S\) and \(x_E\) are trivial seed points, and the vertices of each \(\mathcal{P}_i\) are included as seed points because, by convexity of \(\mathcal{P}_i\), whenever \(\mathcal{D}_{i,j} \neq \emptyset\) some vertex \(\nu\) satisfies \(\mathcal{B}(\nu,Q) \cap \mathcal{E}_{i,j} \neq \emptyset\). An illustration of the wavefront propagation point generator is shown in Fig.~\ref{fig:WavePorpagation}, and Table~\ref{tb:wavepropagation} shows a pseudocode for the wavefront method.

            \begin{figure}
                \begin{center}
                    \includegraphics[width=8.4cm]{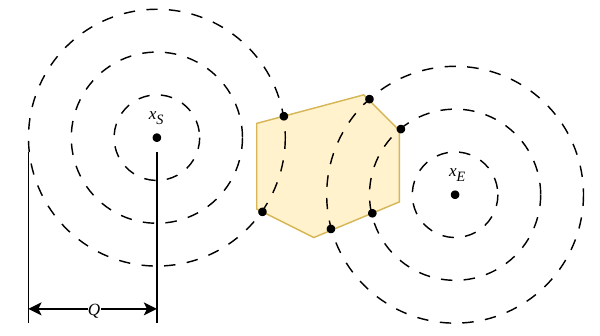}    
                    \caption{An example of the wavefront point generator (with \(\delta=3\)), where points are determined by the intersections of propagated circles up to radius \(Q\) centered at initial (seed) points and the region \(\mathcal{P}\).} 
                    \label{fig:WavePorpagation}
                \end{center}
            \end{figure}

            \begin{table}[hb]
                \begin{center}
                    \caption{Pseudocode of the wavefront method to generate candidate points}
                    \label{tb:wavepropagation}
                    \begin{tabular}{rl}
                        \hline
                        \multicolumn{2}{l}{\textbf{Algorithm 1: Wavefront candidate point generator}} \\
                        \hline
                        1: & \textbf{Inputs} seed points (\(x_S\), \(x_E\), vertices of \(\{\mathcal{P}_1,\dots,\mathcal{P}_m\}\)), \(Q\) \\
                        2: & \textbf{Set} nodes = seed points, step size = \(\delta\) \\
                        3: & \textbf{for} each level \(\Delta_i\) from \(\Delta_1\) to \(\Delta_\delta\) \\
                        4: & \(\quad\) \textbf{for} each seed point \(p_j\) \\
                        5: & \(\quad\) \(\quad\) \textbf{for} each polytope \(\mathcal{P}_i\) in \(\{\mathcal{P}_1,\dots,\mathcal{P}_m\}\)) and \(p_j \notin \mathcal{P}_i\) \\
                        6: & \(\quad\) \(\quad\) \(\quad\) Calculate circle-polygon intersections \\
                        7: & \(\quad\) \(\quad\) \(\quad\) Add intersection points to node set \\
                        8: & \(\quad\) \(\quad\) \textbf{end} \\
                        9: & \(\quad\) \textbf{end} \\
                        10: & \textbf{end} \\
                        11: & \textbf{Return} nodes \\
                        \hline
                    \end{tabular}
                \end{center}
            \end{table}
            
            An edge \(E\) is a nonnegative weight that represents the distance between two nodes. As such, we define the weights of all edges as \(w(u,v) = \|u-v\|\) for any nodes \(u,v \in V\). An edge \(E\) exists if \(\|u - v\| \leq Q\), otherwise the value of the edge is set to \(w(u,v) = \infty\). However, if \(u, v \in \mathcal{P}_i\), i.e., cross polygon interior segments, the edge has unconstrained value. The shortest resource-constrained graph problem for \(\mathcal{G}\) is formulated as
            \begin{equation}
                \label{eq:ShortestGraph}
                \min_{\mathcal{G}} \space \sum_{u,v \in V} w(u,v)
            \end{equation}
            To solve (\ref{eq:ShortestGraph}), we utilize Dijkstra’s algorithm \citep{Candra:2020}, and the algorithm for obtaining the shortest graph is detailed as follows:
            \begin{description}
                \item[\textit{Step 1}] (Candidate Point Generation) Select candidate points \(x_S\), \(x_E\), vertices of \(\mathcal{P}_i\), and the points generated from the wavefront propagation strategy that serves as the nodes of \(\mathcal{G}\).
                \item[\textit{Step 2}] (Graph Construction) Compute the weights \(w\) by calculating all pairwise distances between the nodes while performing feasibility checks.
                \item[\textit{Step 3}] (Shortest Path) Run Dijkstra's algorithm from the source \(x_S\) to the terminal \(x_E\) using the obtained graph.
                \item[\textit{Step 4}] (Path Reconstruction) Recover the sequence of nodes with the minimum values and construct a geometric path. This reconstructed path is the solution to the shortest graph.
            \end{description}

            The step size \(\delta\) controls the granularity of the candidate node set \(V\) and, in turn, the near-optimality of the path length \(L(f)\) and the recovered sequence \(\sigma\). Assuming the existence of \(\mathcal{R}\), as identified in Corollary~3, the algorithm returns a feasible path on \(\mathcal{G}\) and a sequence \(\sigma\) using \(x_S\), \(x_E\) and the vertices of the polytopes as seed points. However, the discretized path length \(L(\mathcal{G})\) can exceed the true optimum \(L^*(f)\) as the nodes might not be placed near the optimal entry/exit points. By increasing \(\delta\), the candidate node set \(V\) covers more points in \(\mathcal{E}_{i,j}\), but this increases the computational complexity of the graph algorithm.

            While the discrete graph search identifies the optimal sequence of regions \(\sigma\), the resulting path is limited by the resolution of the wavefront points. To alleviate this, we apply convex programming subsequently as a continuous-refinement stage to treat the entry/exit points as continuous decision variables, effectively polishing the path to its true global optimum within that sequence.

        \subsection{Solution Refinement with Convex Optimization}

            Since the shortest path that passes a sequence of points consists of piecewise line segments, we can describe the path as
            \begin{multline}
                \label{eq:ShortestLineSegments}
                L(f) = \|x_S - a_{\sigma_1} \| + \|x_E - b_{\sigma_k}\| \\+ \sum_{i=1}^k \|a_{\sigma_i} - b_{\sigma_i}\| + \sum_{i=1}^{k-1} \|b_{\sigma_i} - a_{\sigma_{i+1}}\|
            \end{multline}
            which is a convex function with a global minimum existing. Consequently, the shortest path problem in~(\ref{eq:MinimumPath2}) can be rewritten as
            \begin{subequations}
                \label{eq:ConvexProgramming}
                \begin{align}
                    \min_{\substack{a_{\sigma_i}, b_{\sigma_i} \in \mathcal{P}_{\sigma_i} \\ 1 \leq i \leq k}} \quad & (\ref{eq:ShortestLineSegments}) \\
                    \quad \mathrm{st.} \quad &\|x_S - a_{\sigma_1}\| \leq Q \\
                    \qquad \quad &\|b_{\sigma_i} - a_{\sigma_{i+1}}\| \leq Q, \quad \forall i = 1,\dots,k-1 \\
                    \qquad \quad &\|b_{\sigma_k} - x_E\| \leq Q \\
                    \qquad \quad &H_{\sigma_i} a_{\sigma_i} \leq h_{\sigma_i}, \quad \forall i \in \{1, \dots, k\} \\
                    \qquad \quad &H_{\sigma_i} b_{\sigma_i} \leq h_{\sigma_i}, \quad \forall i \in \{1, \dots, k\}
            \end{align}
            \end{subequations}
            
            Equations~(\ref{eq:ConvexProgramming}b) and (\ref{eq:ConvexProgramming}c) are inequalities to enforce the feasibility of the path to satisfy the length budget \(Q\) outside \(\mathcal{P}_i\), whereas equations (\ref{eq:ConvexProgramming}d) and (\ref{eq:ConvexProgramming}e) are inequalities that ensure the entry and exit points \(a_i\) and \(b_i\) are within \(\mathcal{P}_i\). Since the objective function in~(\ref{eq:ConvexProgramming}) is convex and the constraints are convex and affine, the optimization problem can be solved with convex programming.

            Although~(\ref{eq:ConvexProgramming}) can be solved with convex programming, it scales poorly if the number of replenishment regions increases as this requires solving the convex optimization for all possible sequences \(\sigma\) the path may take; thus \(k\) is a decision variable which increases the computational complexity in return. To avoid this, the resource-constrained graph solution (as shown in Section 3.2) finds the optimal path sequence \(\sigma^*\) the path will make to pass the required replenishment regions to maintain the budget constraint. The convex optimization problem~(\ref{eq:ConvexProgramming}) is solved once only for the given sequence \(\sigma^*\) obtained by Dijkstra's method. As a result, the resource-constrained shortest path problem is solved in two steps: the shortest graph with Dijkstra method provides the topology (sequence), while the convex program provides the exact coordinates.

    \section{Illustrative Example}

        In this section, we show an illustrative example of the proposed method. The numerical results are carried out using MATLAB environment. The convex optimization problem is solved using CVX package for MATLAB \citep{cvx}. We will show a 2-dimensional example to ease illustration, but the approach is also applicable for \(n\)-dimensional cases.
        
        Consider a simple model of a vehicle, equipped with a solar panel, that has to traverse an open area. The environment is assumed to be flat, and there exist no obstacles that would block the system's movements. The vehicle starts at the origin point \(x_S = [0 \quad 0]^\top\) and the objective is to reach the target destination located at \(x_E = [18 \quad 14]^\top\) in the shortest path possible. The vehicle's fuel tank has a range of 3 and consequently, it needs to stop at some of replenishment regions \(\mathcal{P}_i\), indexed by \(i=1,2,\dots,15\), where the vehicle can harvest solar energy for refueling and regain its budget. Refuel regions are randomly generated over the environment as depicted by the yellow polytopes in Fig.~\ref{fig:illustration}. The budget is assumed to be reset instantly during recharge. The main focus of the results is to construct the shortest resource constrained path, thus the time to complete the task is open, and, for the purpose of this example, we assume that the vehicle follows the geometric path, which corresponds to the trajectory of a first-order integrator system.

        \begin{figure}[hb]
            \begin{center}
                \includegraphics[width=8.4cm]{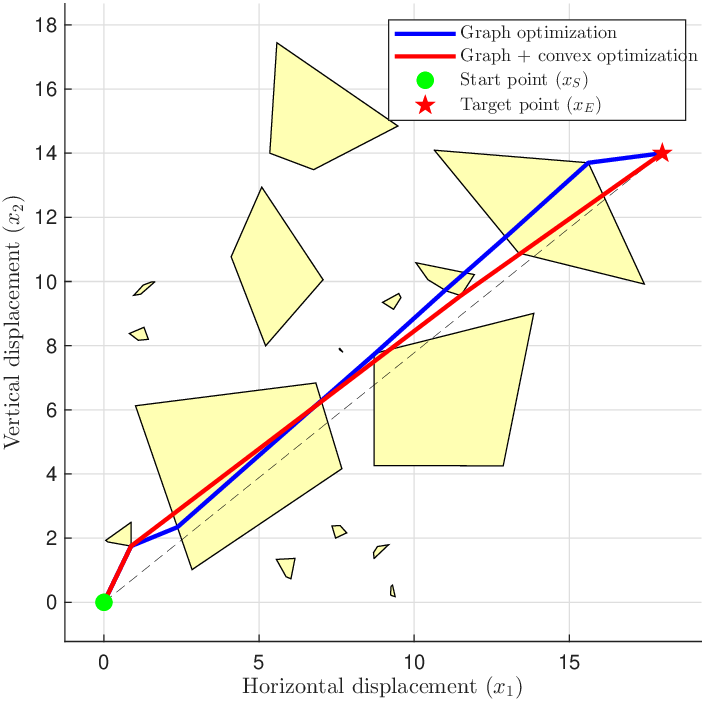}    
                \caption{Illustration of the resource-constrained shortest path in the example.} 
                \label{fig:illustration}
            \end{center}
        \end{figure}

        The result of applying the optimization algorithm is shown in Fig.~\ref{fig:illustration}. The dashed black line is the absolute shortest path from \(x_S\) to \(x_E\) which is equal to \(22.8\) in this case. The algorithm first builds \(\mathcal{G}\) to construct the shortest resource-constrained graph. Here, the nodes are \(x_S\), \(x_E\), the vertices of all \(\mathcal{P}_i\) and the circle-polygon intersections from the wavefront algorithm where the step size of the levels is selected as \(\delta = 4\). Using Dijkstra's method, the algorithm determines the initial calculations of nodes that potentially represent the points \(a_{\sigma_i}\) and \(b_{\sigma_i}\) that decide the shortest feasible polytope sequence \(\sigma\) the path will visit. The blue line in Fig.~\ref{fig:illustration} describes the path constructed solely using the graph method, without further convex programming, where the path only passes the selected candidate nodes and the total length of the path is around \(23.4185\).
        
        Using the sequence \(\sigma\) obtained from solving the shortest path in \(\mathcal{G}\), we now run convex programming on \(\sigma\) to calculate the actual values \(a_{\sigma_i}\) and \(b_{\sigma_i}\) instead of the limited discrete set, improving the result of the shortest path. The red line in Fig.~\ref{fig:illustration} is the shortest path obtained from convex programming on \(\sigma\), and the trajectory of the path is different from the graph method. The total length of the path from the solution of the convex program is approximately \(23.0160\), which is less than the result of the graph method. Note that the result of the graph method can be improved by increasing the candidate space using more levels in the wavefront method, and the results become closer to the convex programming solution. However, the addition of new vertices significantly increases the computational complexity of the algorithm in the process.

    \section{Conclusion}

        This paper considers a framework for solving the continuous resource-constrained shortest path problem where a budget is replenished within convex polytopes. By characterizing the geometric structure of the optimal path, specifically its piecewise linear nature and feasible-connected space property, we reduced a continuous problem into a manageable two-stage graph optimization. We demonstrated that a wavefront-based discretization combined with Dijkstra’s algorithm effectively identifies the optimal sequence of regions, while subsequent convex programming provides the exact optimal entry and exit points. The proposed approach balances computational efficiency with theoretical optimality, outperforming purely discrete graph-based methods. 

        Future work will focus on extending these findings to the optimal control of dynamical systems with internal resource states. This includes designing trajectories that are compatible with higher-order system dynamics and investigating the impact of time-varying resource consumption in the presence of replenishment regions. Additionally, we aim to extend the framework to handle non-convex replenishment regions and dynamic environments with moving obstacles.


    \section*{ACKNOWLEDGMENT}

        The authors would like to thank Ram Padmanabhan for their valuable feedback and editorial assistance throughout the preparation of this manuscript.

    \section*{DECLARATION OF GENERATIVE AI AND AI-ASSISTED TECHNOLOGIES IN THE WRITING PROCESS}
    
        During the preparation of this work the authors used ChatGPT and Google Gemini solely in order to improve of readability and language fluency of the writing. After using this tool/service, the authors reviewed and edited the content as needed and take full responsibility for the content of the publication.
        
    \bibliography{ifacconf}

    
\end{document}